\documentclass[11pt]{article}
\usepackage{amsmath,amsthm,amssymb}
\usepackage[hyperfootnotes=false]{hyperref}
\usepackage{color}
\usepackage{yfonts}
\usepackage{cancel}
\usepackage{titlesec}
\usepackage{enumerate}
\usepackage[numbers]{natbib}
\usepackage[normalem]{ulem}
\usepackage{bm}
\titleformat{\paragraph}
{\normalfont\normalsize\bfseries}{\theparagraph}{1em}{}
\titlespacing*{\paragraph}
{0pt}{3.25ex plus 1ex minus .2ex}{1.5ex plus .2ex}

\def\titlerunning#1{\gdef\titrun{#1}}
\makeatletter
\def\author#1{\gdef\autrun{\def\and{\unskip, }#1}\gdef\@author{#1}}

\makeatother

\def\MSC#1{{\renewcommand{\thefootnote}{}%
\footnote{\emph{Mathematics Subject Classification (2020):} #1}}}
\def\keywords#1{\par\medskip
\noindent\textbf{Keywords:} #1}

\usepackage{graphicx}
\usepackage{longtable}
\usepackage{adjustbox}
\usepackage{multirow}
\usepackage{subfigure}
\usepackage{algorithm}
\usepackage{algorithmic}
\usepackage{changepage}
\usepackage{collcell}
\usepackage{datatool}
\usepackage{lscape}

\usepackage{comment}
\usepackage{booktabs}
\usepackage{tabularx}

\usepackage{amsthm}
\usepackage{thmtools}

\newtheorem{theorem}{Theorem}[section]

\newtheorem{lemma}[theorem]{Lemma}

\newtheorem{proposition}[theorem]{Proposition}
\newtheorem{corollary}[theorem]{Corollary}

\theoremstyle{definition}

\newtheorem{remark}[theorem]{Remark}

\numberwithin{equation}{section}

\def\Aut{{\rm Aut}}
\def\PG{{\rm PG}}

\def\fq{\mathbb{F}_{q}}
\def\fqm{\mathbb{F}_{q^m}}
\def\F{\mathbb{F}}

\def\bx{\bar{x}}

\def\xqs{x^{\sigma}}
\def\xxs{x^{\sigma^2}}
\def\yqs{y^{\sigma}}
\def\yys{y^{\sigma^2}}
\def\yyys{y^{\sigma^3}}
\def\bxqs{\bar{x}^{\sigma}}
\def\bxxs{\bar{x}^{\sigma^2}}
\def\bxxxs{\bar{x}^{\sigma^3}}
\def\lambdas{\lambda^\sigma}
\def\lambdass{\lambda^{\sigma^2}}

\def\ba{\bar{a}}
\def\bb{\bar{b}}
\def\cC{\mathcal{C}}
\def\rk{\mathrm{rk}}

\def\GL{\mathrm{GL}}
\def\GaL{\Gamma \mathrm{L}}

\def \supp{\mathrm{supp}}
\def\cU{\mathcal{U}}

\def\cZ{\mathcal{Z}}
\def\cW{\mathcal{W}}
\newcommand{\C}{\mathcal{C}}

\newcommand{\mU}{\mathcal{U}}

\newcommand{\Fmkd}{[n,k,d]_{q^m/q}}
\DeclareMathOperator{\rowsp}{rowsp}

\usepackage{environ}

\newlength{\myl}
\expandafter\let\expandafter\origequation\csname equation\endcsname
\expandafter\let\expandafter\endorigequation\csname endequation\endcsname
\long\def\[#1\]{\begin{equation}#1\end{equation}}
\RenewEnviron{equation}{
  \settowidth{\myl}{$\displaystyle\BODY$} 
  \origequation
    \ifdim\myl>\linewidth
      \resizebox{\linewidth}{!}{$\displaystyle\BODY$}
    \else
      \BODY 
    \fi
  \endorigequation
}

\NewDocumentEnvironment{env}{+b}
    {\begin{equation}\begin{split}#1\end{split}\end{equation}}
    {}

\begin{document}


\baselineskip=16pt

\titlerunning{}

\title{Short rank-metric codes and scattered subspaces}
\author{ Stefano Lia \and Giovanni Longobardi \and Giuseppe Marino \and Rocco Trombetti}
\date{}

\maketitle


\MSC{94B05 · 51E20 · 94B27}


\begin{abstract}

\noindent By exploiting the connection between scattered $\F_q$-subspaces of $\mathbb{F}_{q^m}^3$ and minimal non degenerate $3$-dimensional rank metric codes of $\mathbb{F}_{q^m}^{n}$, $n \geq m+2$, described in \cite{AlfaranoBorelloNeriRavagnani2022JCTA}, we will exhibit a new class of codes with parameters $[m+2,3,m-2]_{q^m/q}$ for infinite values of $q$ and $m \geq 5$ odd. Moreover, by studying the geometric structures of these scattered subspaces, we determine the rank weight distribution of the associated codes.
\end{abstract}

\section{Introduction}

Codes in the rank metric setting ({\it rank-metric} codes) were originally introduced by Delsarte in the late $70$'s \cite{Delsarte1978bilinear}, and then resumed a few years later by Gabidulin in \cite{Gabidulin1985theory}. The study of these codes has intensified greatly in the last decades, mainly because of their application to network coding \cite{silva2008rank}, but also due to their mathematical structure which link them with many important algebraic and geometric structures in finite geoemtry, such as {\it linear sets} in projective spaces and their underlying $\F_q$-subspaces, \cite{Randrianarisoa2020geometric, SheekeyVoorde2020duality}. 

\medskip
\noindent
Let $m,n \in \mathbb{N}$ be two positive integers such that $m,n \geq 2$, and let $q$ be a prime power. Let $\F_{q^m}^n$ be the vector space of dimension $n$ over the Galois field $\F_{q^m}$. 
Let consider $v=(v_1,v_2,...,v_n) \in \F_{q^m}^n$, the \textit{rank weight} of $v$ is defined as
$$\omega_{\mathrm{rk}}(v)=\dim \langle v_1,v_2,\ldots,v_n\rangle_{\F_q}.$$
A rank-metric code ${\cal C} \subseteq \F_{q^m}^n$ of {\it length} $n$, is a subset of $V$ considered as a vector space endowed with the metric defined by the map $$d(v,w)= \omega_{rk}(v-w),$$
where $v$ and $w \in V.$ Elements of ${\cal C}$ are called {\it codewords}. For any codeword $v \in {\cal C}$ with  $\omega_{\mathrm{rk}}(v)=r$, by fixing an $\F_q$-basis $u=(u_1,\ldots,u_r)$ of $\langle v_1,\ldots,v_n\rangle_{\fq}$, it is easy to see that there exists a matrix $A\in \F_q^{r\times n}$ of maximum rank such that $v=uA$. Moreover, the {\it rank support} of $v$ is defined as
$$\supp_{\rk}(v):=\mathrm{rowsp}(A)\subseteq \F_q^n,$$
where $\mathrm{rowsp}(A)$ denotes the $\F_q$-span of the rows of $A$. By using standard linear algebra, it can be easily seen that the support of a codeword does not depend on the choice of the basis $u$.

The \textit{support} of $\cC$ is simply the span of the rank supports of all its codewords, i.e.
\begin{equation*}
\supp_{\rk}(\cC)=\langle \supp_{\rk}(v) \colon v \in \cC \rangle_{\F_q}.
\end{equation*}
The code $\cC$ is called \textit{(rank-)nondenegerate} if its support is equal to $\F^n_q$.

In this paper we will focus on rank-metric codes which are $\F_{q^m}$-subspaces of $\F_{q^m}^n$, the so-called  {\it $\mathbb{F}_{q^m}$-linear}, or simply {\it linear} rank-metric code. If ${\cal C} \subseteq \F_{q^m}^n$ is a linear rank-metric code, the minimum distance between two distinct codewords of ${\cal C}$ is $$d=d({\cal C})=\min \{\omega_{rk}(v) \, | \, v \in \cal C \setminus \{{\bf 0}\}\}.$$ 

\noindent A linear rank-metric code ${\cal C} \subseteq \F_{q^m}^n$ of length $n$, dimension $k$ and minimum distance $d$, is referred in the literature as to an $[n,k,d]_{q^m/q}$ code, or as to an $[n,k]_{q^m/q}$ code, depending on whether the minimum distance is known or not.  
These parameters are related by an inequality, which is known as the Singleton-like bound. Precisely, if $\cC$ is an $[n,k,d]_{q^m/q}$ code, then 
\begin{equation}\label{singleton-bound}
    mk \leq \min\{m(n - d + 1), n(m - d + 1)\}, 
\end{equation}
see \cite{Delsarte1978bilinear}.
Codes attaining this bound with equality are called \textit{maximum rank distance (MRD) codes}, and they are considered to be optimal, due to their largest possible error-correction capability.\\
A nonzero codeword $v\in\mathcal{C}$ is \textit{minimal} if for every $u \in \mathcal{C}$,
$$ \supp_{\rk}(u)\subseteq \supp_{\rk}(v)  \Longleftrightarrow u=\lambda v, \mbox{ for some } \lambda \in \F_{q^m}.$$
  
\noindent Rank-metric codes for which any nonzero codeword is minimal were introduced in \cite{AlfaranoBorelloNeriRavagnani2022JCTA} with the name of {\it minimal} rank-metric codes, in analogy with the theory of codes in the Hamming metric. These codes have interesting combinatorial and geometric properties besides important applications, for instance to the theory of secret sharing schemes, as shown by Massey in {\cite{massey1993minimal,massey1995some}}. 

Constructing minimal rank-metric codes is not an easy task particularly when assuming short length and low dimension. However, exhibiting new minimal code with short length can be an useful tool towards the construction of longer examples, since the property of being minimal is preserved when extending the code.

In \cite{AlfaranoBorelloNeriRavagnani2022JCTA}, the authors investigated minimal $[n,3]_{q^m/q}$ codes, by connecting them to a class of \emph{scattered subspaces} in a $3$-dimensional $\F_{q^m}$-space; indeed a special type of $q$-{\it systems}, which may be seen as a $q$-analog of the concept of  classical {\it cutting blocking sets}.  Also, in \cite[Corollary 5.10]{AlfaranoBorelloNeriRavagnani2022JCTA} it is proven that the shortest minimal codes with $k=3$ have length $n \geq m+2$. In Theorem 6.7 of the same article, existence of such extremal examples is  established, at least for some parameters, whenever $m\geq 4$, proving that in those cases the theoretical bound is sharp. 
For $m$ odd, an existence result is also given for every $n\geq m + 2$, but only under the further hypothesis that $m\not\equiv 3, 5 \pmod 6$.\\

In regard to these latter achievements, we underline here that they are not constructive in nature. At the best of our knowledge, the unique explicit construction attaining above mentioned lower bound known so far was exhibited in \cite{BartoliCsajbokMarinoTrombetti2021} for $m=3$ and any $q$, and for $m=5$ and some values of the characteristic , namely for $q=p^e$ with $p\in\{2,3,5\}$. The latter examples give rise to {\it maximum scattered subspaces} in $ \F_{q^5}^3$ which are the only known examples of such objects in $\F_{q^m}^k$, in the case when $km$ odd, and $(k,m)\ne (3,3)$.

The existence of minimal rank codes has been faced with recently also in \cite{BMN} and \cite{GruicaRavagnaniSheekeyZullo2022}. In particular, generalizing \cite[Theorem 6.3]{AlfaranoBorelloNeriRavagnani2022JCTA}, in \cite[Theorem 6.3]{BMN} a family of $[n,k]_{q^4/q}$ minimal rank-metric codes of shortest length in the case when $(k,m)=(4,3)$ and $q$ is an odd power of $2$, were constructed. Also, in \cite[Theorem 7.16]{GruicaRavagnaniSheekeyZullo2022} minimal $[m+3,3]_{q^m/q}$ examples have been exhibited for all $m\geq 4$.

\medskip

\noindent 
\noindent In this article, we provide a new class of minimal $[m+2,3]_{q^m/q}$ codes existing for infinite values of $q$ and $m$ having maximum possible value for the minimum distance as well as  the largest possible value for the second generalized weight.
As existence conditions for such codes, we require some arithmetic conditions involving $q$ and $m$ and that a certain polynomial in $\F_q[X]$ has no roots in $\F_{q^m}\setminus \F_q$ (cf. Theorem \ref{MainTheorem} for the precise statement and subsequent results). These codes are obtained by constructing a class of $(m+2)$-dimensional scattered $\F_q$-linear sets of the projective plane $\PG(2,q^m)$. These are interesting in their own right, as they turn out to have three intersection characters with respect to lines of $\PG(2,q^m)$. When $m=5$, this shows that maximum scattered linear sets of $\PG(k-1,q^m)$, when $km$ is odd, have a  different combinatorial behaviour in comparison with maximum linear sets in the case when $km$ is even. Indeed, in the latter case they have two intersection characters (cf. \cite{blokhuis2000scattered}). 

The cases $m=5$ and $m=7$, are studied in more details; indeed, in these cases we provide sufficient conditions in order for the above mentioned linear sets to be cutting blocking sets of minimal dimension. Finally, we show the existence of a family of maximum scattered subspaces in $\F_{q^5}^3$ for any $q=2^{2h+1}$, $h\geq 1$, and $q\equiv 2,3 \pmod 5$ odd. 

\medskip
\noindent\textbf{Outline.} The paper is structured as follows. Section \ref{sec:preliminaries} introduces the main objects and notions that we need in the paper, giving a brief recap on rank-metric codes, $q$-systems, their evasiveness and cutting properties. In Section \ref{sec:construction} we construct a class of $[m+2,3]_{q^m/q}$ minimal rank codes. The cases $m=5$ and $m=7$ are further investigated in Subsections \ref{subsec:Case5} and \ref{subsec:Case7}. 
In Section \ref{sec:equivalence} we deal with the equivalence issue.
In Section \ref{sec:characters} we study the intersections between the corresponding $\F_q$-linear sets of $\PG(2,q^m)$ and the lines.
We conclude in Section \ref{sec:further results}, listing some open problems and new research directions.

\section{Preliminaries}\label{sec:preliminaries}

In this section, we will give some details on the notions of $q$-systems, evasive (and scattered) subspaces, and linear cutting blocking sets of $\F_{q^m}^k$, elaborating on the connections among these geometric objects, as well as on the close link they have with minimal linear rank-metric codes defined in the previous section.

A \textit{generator matrix} of an $[n,k]_{q^m/q}$ code ${\cal C} \subseteq \F_{q^m}^n$, is a matrix $G=(g^T_1 \mid g^T_2 \mid \cdots \mid g^T_n) \in \F^{k \times n}_{q^m}$, $g_i \in \F_{q^m}^k$, whose rows span the code as an $\F_{q^m}$-linear space. 

The dual code of $\cC$ is defined to be:
\begin{equation*}
   \cC^\perp
= \{v \in \F_{q^m}^n
\colon v \cdot c = 0\,\,\forall c \in  \cC\}, 
\end{equation*}
where $\cdot$ is the standard scalar product of $\F^n_{q^m}$.

Finally, concerning with equivalence of codes there are few ways to introduce this notion. Here, we only consider equivalences of codes given by linear isometries of the ambient space $\F_{q^m}^n$ (for more details see \cite{Berger}). More precisely, we say that two $[n,k,d]_{q^m/q}$ codes $\cC_1, \cC_2$ are \textit{(linearly) equivalent} if there exist $A\in \GL(n,q)$ and $a\in\F_{q^m}^*$ such that 
\begin{equation*}
\cC_2=a\cC_1\cdot A=\{avA \,:\, v \in \cC_1\}.
\end{equation*}
The class of all (linearly) equivalent $[n,k,d]_{q^m/q}$ codes will be indicated by $\textfrak{C}[n,k,d]_{q^m/q}$.

\subsection{$q$-Systems, evasive and cutting subspaces}

In \cite{Randrianarisoa2020geometric},  linear non degenerate $[n,k]_{q^m/q}$ codes were associated with a class of $n$-dimensional $\F_q$-subspaces of $\F_{q^m}^k$. These subspaces are named there $q$-systems, and they constitute a natural generalization, within the rank metric context, of the concept of {\it projective system} associated with a linear {\it non-degenerate} Hamming code. 

The notion of $q$-system was first introduced in \cite{SheekeyVoorde2020duality}, in a different and more general geometric setting. However, in the following we recall the formulation appearing in \cite{Randrianarisoa2020geometric}, which is more consistent with recent literature on the topic, and with the notation we use in the present article. 

Let $\cU$ be an $\F_q$-subspace of $\F_{q^m}^k$ and let $H$ be an $\F_{q^m}$-subspace of $\F_{q^m}^k$. The \textit{weight} of $H$ in $\cU$ is $\mathrm{wt}_{\cU}(H)=\dim_{\F_q}(H\cap \cU)$. Assume now that $\cU$ has dimension $n$ over $\F_q$. 

We say that $\cU$ is a $[n,k,d]_{q^m/q}$ \textit{system} if $\langle \cU \rangle_{\F_{q^m}}=\F_{q^m}^k$ and 
\begin{equation*}
d=\,n-\mathrm{max}\{\mathrm{wt}_{\cU}(H) \,:\, H \subseteq \F_{q^m}^k \textnormal{ with } \dim_{\F_{q^m}}(H)=k-1 \}.
\end{equation*}

More generally, for each $1 \leq \rho \leq k-1$, the parameters

\begin{equation*}
d_{\rho}=\,n-\mathrm{max}\{\mathrm{wt}_{\cU}(H) \,:\, H \subseteq \F_{q^m}^k \textnormal{ with } \dim_{\F_{q^m}}(H)=k-\rho \},
\end{equation*}
are known as the {\it $\rho$-generalized rank weight} of the system $\cU$, see \cite[Definition 4]{Randrianarisoa2020geometric}.

 As before, if the parameter $d$ is not relevant, we will write that $\cU$ is an $[n,k]_{q^m/q}$ system. Furthermore, when none of the parameters is relevant, we will generically refer to $\cU$ as to a $q$-system.

 Two $[n,k,d]_{q^m/q}$ systems $\cU_1,\cU_2$ are (linearly) equivalent if there exists $A\in \GL(k,q^m)$ such that 
 \begin{equation*}
 \cU_1\cdot A:=\{uA \,:\, u \in \cU_1\}=\cU_2.
 \end{equation*}

The class of all (linearly) equivalent $[n,k,d]_{q^m/q}$ systems will be indicated by $\textfrak{U}[n,k,d]_{q^m/q}.$ 

In \cite{Randrianarisoa2020geometric}, a one-to-one correspondence between the equivalence class of non-degenerate codes $\textfrak{C}[n,k,d]_{q^m/q}$ and the equivalence class of the $q$-systems $\textfrak{U}[n,k,d]_{q^m/q}$, is established. Precisely, once defined the following maps:

$$\begin{array}{rccc}\Phi: & \textfrak C(n,k,d)_{q^m/q} &\longrightarrow &\textfrak U(n,k,d)_{q^m/q} \\
& [\rowsp( g_1^T \mid \ldots \mid g_n^T)] & \longmapsto & [\langle g_1, \ldots, g_n\rangle_{\mathbb{F}_q}] \end{array} $$

$$\begin{array}{rccc}\Psi: & \textfrak U(n,k,d)_{q^m/q} &\longrightarrow &\textfrak C(n,k,d)_{q^m/q} ,\\
& [\langle g_1, \ldots, g_n \rangle_{\mathbb{F}_q}] & \longmapsto & [\rowsp( g_1^T \mid \ldots \mid g_n^T)] \end{array},$$

the following result is stated.

\begin{theorem}\label{th:correspondence}
 The maps $\Phi$ and $\Psi$ are well-defined and they are
the inverses of each other. Hence, they define a 1-to-1 correspondence between equivalence classes of $\Fmkd$ codes and equivalence classes of $\Fmkd$ systems. Moreover, under this correspondence that associates a $\Fmkd$ code $\C$ to an $\Fmkd$ system $\mU$, codewords of $\C$ of rank weight $w$ correspond to $\mathbb{F}_{q^m}$-hyperplanes $H$ of $\mathbb{F}_{q^m}^k$ with $\mathrm{wt}_\cU(H)=n-w$.
\end{theorem}

Moreover, if $\C$ is the $\Fmkd$ code associated to $\mU$, the $\rho$-generalized weight of $\C$ is simply defined as the $\rho$-generalized weight of one associated system $\mU$ in $\textfrak U(n,k,d)_{q^m/q}$, see \cite[Definition 5]{Randrianarisoa2020geometric}. It is easy to see that the first generalized weight $d_1$ of $\C$, coincides with its minimum distance $d$.

Let $h,r \in \mathbb{N}$ be two positive integers such that $h<k$ and $r < km$. An $[n,k]_{q^m/q}$ system $\cU \in \F_{q^m}^k$ is said to be an $(h,r)_q$-\textit{evasive subspace} (or simply \textit{$(h,r)_q$-evasive}) if $\mathrm{wt}_{\cU}(H) \leq r$ for each $\F_{q^m}$-subspace $H$ of $\mathbb{F}_{q^m}^k$ with $\dim_{\F_{q^m}}(H)=h$.
When $r=h$, an $(h,h)_q$-evasive subspace is called \textit{$h$-scattered}. Furthermore, when $h=1$, a $1$-scattered subspace will be called simply a \textit{scattered} subspace. In \cite[Theorem 2.3]{CMPZ2021-generalising}, it was proven that if $\cU$ is an $h$-scattered subspace of $\F_{q^m}^k$, then either $\dim_{\F_q}\cU=k$ and $\cU$ defines a subgeometry of $\PG(k-1,q^m)$ (and it is $(k-1)$-scattered) or $\dim_{\F_q}(\cU)\leq \left \lfloor \frac{km}{h+1} \right \rfloor$. When the latter equality is attained, then $\cU$ is said to be \textit{maximum $h$-scattered}.

Evasive subspaces are a special family of evasive sets, which were introduced first by Pudl\'ak and R\"odl \cite{PudlakVojtvech2004}. These objects were then exploited by Guruswami \cite{guruswami2011linear, GuruswamiWangXing2016}, Dvir and Lovett in \cite{dvir2012subspace} as a tool to construct codes with optimal rate and good list decodability performance.
For more details on this topic, we refer the reader to \cite{BartoliCsajbokMarinoTrombetti2021}, where a mathematical theory of evasive subspaces was developed.

We conclude this section by recalling another family of $q$-systems which was recently introduced in \cite{AlfaranoBorelloNeriRavagnani2022JCTA}.

An $[n,k]_{q^m/q}$ system $\cU$ is said to be $t$-\textit{cutting}, if for every $\F_{q^m}$-subspace $H$ of $\F_{q^m}^k$ of co-dimension $t$, we have $\langle H \cap \cU \rangle_{\F_{q^m}} =H$. When $t=1$, we simply say that $\cU$ is a \textit{cutting} (or a \textit{linear cutting blocking set}). These objects are strictly related to minimal rank-metric codes, indeed as showed in \cite[Corollary 5.7]{AlfaranoBorelloNeriRavagnani2022JCTA}, an $[n,k]_{q^m/q}$ code $\cC$ is minimal 
if and only if the associated $[n,k]_{q^m/q}$ systems $\cU$ is a linear cutting blocking set. 

Since every system containing a linear cutting blocking set is itself a linear cutting blocking set, when aiming at the construction of new examples, it is natural to look for small ones; or, in other words, for the existence of short minimal rank-metric codes.

As a code theoretic counterpart, in \cite[Corollary 5.8]{AlfaranoBorelloNeriRavagnani2022JCTA}, it is shown that a minimal $[n,k]_{q^m/q}$ rank-metric code can be always extended to a minimal $[n+1,k]_{q^m/q }$ code. For this reason, from now on we focus on short minimal codes with parameters $[m+2,3]_{q^m/q}$. Indeed in \cite{AlfaranoBorelloNeriRavagnani2022JCTA}, exploiting the above mentioned link, the authors were able to derive a lower bound on the dimension of a $[n,k]_{q^m/q}$ linear cutting blocking set. Precisely, in \cite[Corollary 5.10]{AlfaranoBorelloNeriRavagnani2022JCTA}, it is proven that if $\cU$ is a cutting $[n,k]_{q^m/q}$ system, with $k \geq 2$, then $n\geq m+k-1$.
In particular, when $k=3$ a scattered $[n,3]_{q^m/q}$ system, with $n \geq m+2$, turns to be a linear cutting blocking set, \cite[Theorem 6.3]{AlfaranoBorelloNeriRavagnani2022JCTA}. This result has been generalized as follows.

\begin{theorem}\cite[Theorem  3.3]{BMN}\label{evasive iff cutting}
    Let $\cU$ be an $[n, k]_{q^m/q}$ system. Then, $\cU$ is $(k - 2, n -m - 1)_q$-evasive if and only if it is cutting.
\end{theorem}

As a consequence of the theorem above, if $m<(k-1)^2$ then there are no linear cutting $[m+k-1,k]_{q^m/q}$ systems, \cite[Corollary 3.5]{BMN}.

The following result provides a quite natural generalization of \cite[Theorem 7.5]{GruicaRavagnaniSheekeyZullo2022} and  \cite[Theorem 3.3]{BMN} to the case of evasive subspaces.  
\begin{theorem}\label{th:evasiveness_vs_cutting}
    Let $\cU$ be an $[n, k]_{q^m/q}$ system with $n>m(k-t)-1$. Then $\cU$ is a $(k-t)$-cutting if and only if it is a $(t-1,n-(k-t)m-1)_q$-evasive subspace of $\mathbb{F}_{q^m}^k$.
\end{theorem}
\begin{proof}
Let $h:=n-(k-t)m-1$. Firstly, let assume that $\cU$ is a $(k-t)$-cutting subspace of $\F^k_{q^m}$ of dimension $n$.
Since $\cU$ is $(k-t)$-cutting, every $t$-dimensional subspaces through $M$ has $q^t>q^h$ vectors in $\cU$. It follows that 
\begin{equation*}
\begin{split}
q^n &=| \cU|  \geq \sum_{\overset{M \leq L}{ \dim L =t}} \left (|L \cap \cU | \setminus |M \cap \cU| \right )+ |M \cap \cU| \geq (q^t-q^h)\frac{q^{m(k-t)+1}-1}{q^m-1}+q^h \\
& \geq q^{h}(q-1)\frac{q^{m(k-t)+1}-1}{q^m-1}+q^h
\geq q^{n-m(k-t)}(q-1)\frac{q^{m(k-t+1)}-1}{q^m-1}+q^{n-m(k-t)}\\
&\geq q^{n-m(k-t)}\left (\frac{q^{m(k-t+1)}-1}{q^m-1}+1\right )=q^{n-m(k-t)}(q^{m(k-t)}+q^{m(k-t-1)}+\ldots +2) >q^n
\end{split}
\end{equation*} 
a contradiction.\\
On the other hand, let $W$ be an $\F_{q^m}$-subspace of $\F^k_{q^m}$ of dimension $t$. Then
\[\mathrm{wt}_\cU(W)=n+tm-\dim_{\F_q}(\cU+W)\geq n+tm-km=n-m(k-t)\]
    If $\langle W\cap \cU\rangle_{\F_{q^m}}\subsetneq W$ then there would be an $\F_{q^m}$-subspace, say $\overline W$, of $\mathbb{F}_{q^m}^k$ of dimension $t-1$ such that $\langle W\cap \cU\rangle_{\F_{q^m}}\subseteq \overline{W}\subset W$. Since $n-m(k-t)-1 \geq \mathrm{wt}_{\cU}(\overline W)=\mathrm{wt}_{\cU}(W)\geq n-m(k-t)$, one gets a contradiction.
\end{proof}


\section{Construction of a class of \texorpdfstring{$[m+2,3]_{q^m/q}$}{} systems, $m$ odd}
\label{sec:construction}

 As a consequence of what explained in the previous section, the shortest minimal $3$-dimensional rank-metric codes over $\F_{q^m}$, are $[m+2,3]_{q^m/q}$ codes.  Also, their existence reduces to the existence of $(m+2)$-dimensional scattered subspaces of $\F_{q^m}^3$, which in turn are $(m+2)$-dimensional linear cutting blocking subspaces of $\F_{q^m}^3$. 

In this section, we will construct a new infinite family of such objects for $m\geq 5$ odd. Note that for $m$ even such objects always exist \cite{ball2000linear, bartoli2018maximum,blokhuis2000scattered,  csajbok2017maximum}, while for $m\leq 3$ they do not exist, \cite{blokhuis2000scattered}. We first recall some concepts and tools from the theory of $q$-linearized polynomials, along with some results on the number of solutions of equations associated with them, taken from \cite{SheekeyMcGuire2019}.

Let $q=p^e$ with $p$ prime and let  $\sigma: x \in \fqm \longrightarrow x^{q^s} \in \fqm$ be a field automorphism of $\fqm$ with $1 \leq s \leq m-1$ and $\gcd(s,m)=1$.
A $\sigma$-\textit{linearized polynomial} with coefficients over $\F_{q^m}$ is a polynomial of the form
\begin{equation}\label{def:sigma_linearized_poly}
L(X) = \sum^r_{i=0} \alpha_iX^{\sigma^i}
\in \F_{q^m}[X], \quad r \in \mathbb{N}
\end{equation}
If $r$ is the largest integer such that $\alpha_r\not = 0$, then it will be called $\sigma$-\textit{degree} (or $q^s$-\textit{degree}) of $L(X)$, in symbols $\deg_{\sigma}(L)= r$. 
The $\sigma$-linearized polynomials define endomorphisms of $\F_{q^m}$ seen as a vector space over $\F_q$. It is well-known that 

$$\mathcal{\tilde{L}}_{m,q}[X] =\Bigl  \{\sum_{i=0}^{m-1}\alpha_iX^{\sigma^i} : \alpha_i \in \F_{q^m}, i \in \{0,1,\ldots,m-1\}\Bigr \}$$
can be endowed with an algebra structure $(\mathcal{\tilde{L}}_{m,q}[X],+,\circ,\cdot)$ where $+$ is the addition of polynomials, $\circ$  the composition of polynomials$\pmod{X^{\sigma^m}-X}$ and $\cdot$ the scalar multiplication by elements of $\F_{q}$, and this algebra is isomorphic to $\mathrm{End}_{\F_q}(\F_{q^m})$.

A polynomial of the form
\begin{equation}\label{eq:projective_ppoly}
P(X)=\sum^r_{i=0}
\alpha_i X^{\frac{\sigma^i-1}{\sigma-1}} \in \F_{q^m}[X], \quad r \in \mathbb{N},
\end{equation}
will be called $\sigma$-\textit{projective} polynomial with coefficients over $\F_{q^m}$. We say that $r$ is the $\sigma$-degree of the polynomial in (\ref{eq:projective_ppoly}).

Projective polynomials were introduced by Abhyankar in \cite{Abhyankar}, and have been studied over finite fields for example in \cite{Bluher2004} and in \cite{Collisions}.\\
Now, let  $L(X)=\sum^r_{i=0}\alpha_iX^{\sigma^i}$ be a $\sigma$-linearized polynomial with $\sigma$-degree $r$. Then, it is possible to associate with $L(X)$ a projective polynomial, i.e. the polynomial $P_L(X)$ obtained by simply substituting in  (\ref{def:sigma_linearized_poly}) the $\sigma$-power $\sigma^i$ by $\frac{\sigma^i-1}{\sigma -1},$ for each $i \in \{0,1,...,r\}$. Note that one has $L(X)=XP_L(X^{\sigma-1})$ and this correspondence with the linearized polynomials is 1-1.

Let $L(X)=\sum^r_{i=0}\alpha_iX^{\sigma^i} \in \tilde{\mathcal{L}}_{m,\sigma}$ with $\deg_{\sigma}(L)=r$. The \textit{companion matrix} of $L(X)$ (and hence also of $P_L(X)$) is the square matrix $C_L$ of order $r$:
\begin{equation*}
C_L = 
\begin{pmatrix}
0 & 0 & \ldots & 0 & -\frac{\alpha_0}{\alpha_r}\\
1 & 0 &\ldots & 0 & -\frac{\alpha_1}{\alpha_r}\\
\vdots & \vdots & \ddots & \vdots & \vdots \\
0 & 0 & \ldots & 1 & -\frac{\alpha_{r-1}}{\alpha_r}
\end{pmatrix}.
\end{equation*}
Consider the matrix
\begin{equation*}
A_L = C_LC_L^\sigma \ldots C_{L}^{\sigma^{m-1}}.
\end{equation*}
In \cite{SheekeyMcGuire2019}, necessary and sufficient
conditions on the
number of roots of a projective polynomial $P_L(X)$ are given, in relation to the eigenvectors of $A_L$.
More precisely
 \begin{theorem}\label{eigen}\cite[Theorem 6]{SheekeyMcGuire2019} The number of roots of $P_L$ in $\F_{q^m}$ equals 
 
\begin{equation*}
\sum_{\lambda \in \F_{q}}
\frac{q^{n_\lambda}- 1}{q - 1} ,
\end{equation*}
where $n_\lambda$ is the dimension of the eigenspace of $A_L$ corresponding to the eigenvalue $\lambda$.
The number of roots of $L(X)$ in $\F_{q^m}$ is equal to $q^{n_1}$, i.e., to the size of the eigenspace of $A_L$
corresponding to the eigenvalue $1$.
 \end{theorem}
Since the matrix $A_L$ is product of $\sigma^i$-th power of $C_L$ with $i \in \{1,\ldots,m-1\}$, it is rather complicated to give an explicit criteria for determining the
number of roots of $L$ and $P_L$ in the case of a general $\sigma$-degree; however, an explicit one is developed in the case $r \in \{2,3\}$, \cite[Section 4 and 5]{SheekeyMcGuire2019}.
Since we will make heavy use of it in subsequent Section \ref{sec:construction}, we will recap this criterion in the case $r=\deg_{\sigma}(L)=2$.

Let $L(X) = \alpha_0x + \alpha_1x^\sigma + \alpha_2x^{\sigma^2}$
with $\alpha_i \in  \F_{q^m}$, $\alpha_0\alpha_2 \not = 0$. Then
\begin{equation*}
 C_L =
\begin{pmatrix}
0 & -\alpha_0/\alpha_2 \\
1 & -\alpha_1/\alpha_2
\end{pmatrix}.
\end{equation*}
Note that it is always possible to assume that 
 $\alpha_0\alpha_2 \not = 0$. Indeed, since $\deg_\sigma(L)=2$, $\alpha_2 \not =0$ and  if $ \alpha_0 = 0$, one can usually apply a power
of $\sigma$ to $L$ to obtain a polynomial that has $\alpha_0 \not = 0$.
Following \cite[Proposition $4$]{SheekeyMcGuire2019}, by putting $u=\frac{\alpha_0^\sigma\alpha_2}{\alpha^{\sigma+1}_1}$, one has 
\begin{equation}\label{matrixAL}
A_L = \mathrm{N}_{q^m/q}(\alpha_1/ \alpha_2)
\begin{pmatrix} 
-u^{\sigma^{-1}}G_{m-2}^{\sigma} & -(\alpha_0/\alpha_1)G^\sigma_{m-1}\\
(\alpha_2/\alpha_1)^{\sigma^{-1}}G_{m-1} & G_m
\end{pmatrix},
\end{equation}
where $G_k$'s  are defined recursively according to the formula:
\begin{equation}\label{eq:recursion}
G_k + G^{\sigma}_{k-1} + uG_{k-2}^{\sigma^2} = 0.
\end{equation}
with $k \in \{2,\ldots,m\}$ and $G_0=1$ and $G_1=-1$.
So, in order to compute the number of roots of $L(X)$ and $P_L(X)$, we need to know the characteristic
polynomial 
$$\chi_L(X)=X^2-\mathrm{Tr}(A_L)X+\det(A_L)$$
of $A_L$. By \eqref{matrixAL}, we get 
\begin{equation*}
\det(A_L) = \mathrm{N}_{q^m/q}(\alpha_1/\alpha_2)^2u^{\sigma^{-1}}
(G_{m-1}^{\sigma+1}- G_mG^\sigma_{m-2})=\mathrm{N}_{q^m/q} (\alpha_0/\alpha_2)
\end{equation*}
and
\begin{equation}\label{eq:description_Tr(A)}
\mathrm{Tr}(A_L) = \mathrm{N}_{q^m/q} (\alpha_1/\alpha_2)(G_{m} - u^{\sigma^{-1}}
G^\sigma_{m-2})= \mathrm{N}_{q^m/q} (\alpha_1/\alpha_2)(G_m + G^\sigma_m+ G^\sigma_{m-1}
).
\end{equation}

By Theorem \ref{eigen}, the number of roots of $P_L(X)$ in $\F_{q^m}$ is determined by the dimension of the
eigenspaces of $A_L$ corresponding to eigenvalues in $\F_q$. Hence, if $\chi_L(X)$ has two distinct
roots in $\F_q$, then $P_L(X)$ has two roots in $\F_{q^m}$. If $\chi_L(X)$ has no roots in $\F_q$, then $P_L(X)$ has no
roots in $\F_{q^m}$. If otherwise $\chi_L(X)$ has a double root in $\F_q$ , then $P_L(X)$ has either $1$ or $q + 1$ root(s) in
$\F_{q^m}$; if $A_L$ is a scalar multiple of the identity then $P_L(X)$ has $q + 1$ roots, otherwise it has 1
root. 

\medskip

\noindent Next, consider the following $\F_q$-linear subspace $\cU_\sigma$ of $\mathbb{F}_{q^m}^3$, $m\geq 5$:
\begin{equation}\label{stsubspace}
    \cU_\sigma=\{(x,\xqs+a,\xxs +b):x \in\fqm,a,b\in\fq\}. 
\end{equation} 

We are now in the position of proving our main contributions; namely, we will give some conditions under which $\cU_\sigma$ is scattered. 

Consider the $2$-dimensional $\fqm$-vector subspaces of $\mathbb{F}_{q^m}^3$ containing the vector $(0,0,1)$, they have equation 
\begin{equation}\label{bundle}
\ell_{\lambda}: x_1=\lambda x_0 \textnormal{ or } \ell_{\infty}: x_0=0
\end{equation}
and define $\mathcal{Z}_{\lambda,\sigma}=\cU_{\sigma} \cap \ell_\lambda$ with $\lambda \in \fqm \cup \{\infty\}$.
Since $\cU_\sigma \cap \langle v \rangle _{\F_{q^m}}=\mathcal{Z}_{\lambda,\sigma} \cap \langle v \rangle_{\F_{q^m}}$ for some $\lambda \in \F_{q^m} \cup \{\infty\}$,  the following result is plain. 
\begin{lemma}\label{local-scatt}
Let $\cU_\sigma$ be the $\F_q$-linear subspace of $\F_{q^m}^3$ as defined in \eqref{stsubspace}. Then $\cU_\sigma$ is a scattered subspace if and only if $\mathcal{Z}_{\lambda,\sigma}$ is scattered for any $\lambda \in \fqm \cup \{\infty\}$.
\end{lemma}

\medskip

Firstly, note that $$\mathcal{U}_\sigma =  \mathcal{W}_\sigma \oplus \cZ_{\infty,\sigma},$$ where
$\mathcal{W}_\sigma=\{(x,\xqs,\xxs ):x \in\fqm\}$, and  $\cZ_{\infty,\sigma}=\langle (0,0,1),(0,1,0)\rangle_{\F_q}$; also, these are scattered subspaces.

Then, let $\lambda \in \fqm$ and $\bar{v}=(\bx,\bxqs+\ba,\bxxs +\bb) \in \cZ_
{\lambda,\sigma}$ with $\bx \in\fqm^*,\ba,\bb\in\fq$ and, hence, $\bx^\sigma+ \ba= \lambda \bx$. 
The property of being scattered for $\cZ_{\lambda,\sigma}$ is equivalent to require that the number of triples $(y,y^\sigma+a,y^{\sigma^2}+b)$ with $y \in \F_{q^m},a,b \in \F_q$
such that 
\begin{equation}\label{sistema}
\begin{cases}
&\lambda y =\yqs+a\\
&\frac{y}{\bx}(\bxqs+\ba)=\yqs+a \\
&\frac{y}{\bx}(\bxxs+\bb)=\yys+b \\
 \end{cases} 
\end{equation}
is at most $q$. Since the values of $a$ and $b$ in \eqref{sistema} can be expressed in term of $y$ by using the second and third equation, one sees that  $\cZ_{\lambda,\sigma}$ is scattered if and only if for any $\bar{v} \in \cZ_{\lambda,\sigma}$, there are at most $q$ values of $y \in \F_{q^m}$ for which System \eqref{sistema} has a solution.

\begin{proposition}
If $\gcd(q-1,m)=1$, $q=p^e$, and $p$ does not divide $m$, then  for any $\lambda \in \fq$  the subspace $\mathcal{Z}_{\lambda,\sigma}$ of $\F_{q^m}^3$ is scattered.
\end{proposition}
\begin{proof}
Let $\lambda$ be an element of $\F_q$ and $\bar{v}=(\bx,\bxqs+\ba,\bxxs +\bb) \in \cZ_
{\lambda,\sigma}$ with $\bx \in\fqm^*,\ba,\bb\in\fq$ and $\bx^\sigma+ \ba= \lambda \bx$. 
If $\lambda=0$, by System \eqref{sistema}, it follows $y\in\fq$, which completes this case. 

Assume $\lambda\in\fq^*$. Since $a\in\fq$, by the first equation in \eqref{sistema}, it follows that
$y^{\sigma}-\lambda y=y^{\sigma^2}-\lambda y^{\sigma}$ or, in other words that,
\begin{equation}\label{q2avoid}
 y^{\sigma^2}-(1+\lambda)y^{\sigma}+\lambda y=0.   
\end{equation} 
By \cite[Theorem 1.2]{CMPZ2019}, this polynomial has exactly $q^2$ solutions in $\F_{q^m}$ if and only if the $m$-th power of the matrix 
\begin{equation*}
    A_\lambda=\begin{pmatrix}
        0 &  -\lambda\\
        1 & 1+\lambda
\end{pmatrix},
\end{equation*}
is equal to the identity matrix $I_2$.
Thus a necessary condition for the equation in \eqref{q2avoid}  to have $q^2$ zeros is that $\lambda^m=\det(A_\lambda)^m=1$ and since $\gcd(q-1,m)=1$, this is possible only if $\lambda=1$.
So, the polynomial map $y^{\sigma^2}-(1+\lambda)y^{\sigma}+\lambda y$ has maximum kernel if and only if $A^m_1=I_2$. 
By a simple inductive argument, one gets that 
\begin{equation*}
A^m_1=
\begin{pmatrix}
   -(m-1) & -m \\
   m & m+1
\end{pmatrix}
\end{equation*}
and hence $A^m_1=I_2$ if and only if $p | m$. This is not possible by the hypothesis.
 Then the polynomial in \eqref{q2avoid} has at most $q$ zeros and this concludes the proof.
\end{proof}

Now, it remains to give conditions in order to show that the subspace  $\cZ_{\lambda,\sigma}$ is scattered for each $\lambda \in \F_{q^m} \setminus \F_q$. By System \eqref{sistema},
since $a=\lambda y - y^\sigma$, $b= \frac{y}{\bx}(\bxxs+\bb)-\yys$ and $\ba,\bb$ belong to $\fq$,
\begin{equation*}
\lambda y -\yqs=\lambdas \yqs-\yys,
\end{equation*}
which implies 
\begin{equation}\label{yquadro}
\yys=(\lambdas+1)\yqs-\lambda y
\end{equation} and 
\begin{equation}\label{equazioneb}
\frac{y}{\bx}(\bxxs+\bb)-\yys=\frac{\yqs}{\bxqs}(\bxxxs+\bb)-\yyys.
\end{equation}
Since $m\geq 3$, substituting the expression  $y^{\sigma^2}$, obtained in \eqref{yquadro}, in \eqref{equazioneb}, one gets
\begin{equation*}
\bigg( (\lambdass+1)(\lambdas+1)-\lambdas-(\lambdas+1)-\frac{\bxxxs+\bb}{\bxqs}\bigg)\yqs+ \bigg(\frac{\bxxs+\bb}{\bx}+\lambda-\lambda(\lambdass+1)\bigg)y=0.
\end{equation*}
This is a $\sigma$-linearized polynomial in $y$, and it has more than $q$ solutions if and only if its coefficients  are both identically zero and so if and only if 
\begin{equation}\label{coeffy}
\begin{cases}
&\frac{\bxxs+\bb}{\bx}+\lambda-\lambda(\lambdass+1)=0 \\
&  (\lambdass+1)(\lambdas+1)-\lambdas-(\lambdas+1)-\frac{\bxxxs+\bb}{\bxqs}=0\\
 \end{cases} 
\end{equation}
By the first equation of\eqref{coeffy}, $\frac{\bxxs+\bb}{\bx}=\lambda(\lambdass+1)-\lambda$ holds. So, a necessary condition for the equation to admit $q^2$ zeros is that  
\begin{equation*}
\lambda^{\sigma^3+\sigma}-\lambda^{\sigma^2+\sigma}-\lambda^{\sigma^2}+\lambdas=0;
\end{equation*}
which is equivalent to 
\begin{equation*}
\lambda^{\sigma^2+1}-\lambda^{\sigma+1}-\lambda^{\sigma}+\lambda=0.
\end{equation*}
Therefore, we have the following.

\begin{proposition}
Let  $\lambda \in \F_{q^m} \setminus \F_q$, $m \geq 5$. If $\lambda$ is not a root of the polynomial
\begin{equation}\label{polynomial}
  Q(X)=  X^{\sigma^2+1}-X^{\sigma+1}
-X^\sigma+X \in \F_{q}[X],
\end{equation}
then  the subspace $\cZ_{\lambda,\sigma}$ of $\F_{q^m}^3$ is scattered.
\end{proposition}
Summing up the results so far, by Lemma \ref{local-scatt}, we get
\begin{theorem}\label{MainTheorem}
Let $q=p^e$ and $m \geq 5$. Consider the $(m+2)$-dimensional subspace
\begin{equation*}
    \cU_\sigma=\{(x,\xqs+a,\xxs +b):x \in\fqm,a,b\in\fq\}
\end{equation*}
of $\F_{q^m}^3$, where $\sigma: x \in \F_{q^m} \longrightarrow x^{q^s} \in \F_{q^m}$, $1 \leq s \leq m-1$ and  $\gcd(s,m)=1$.
If 
\begin{itemize}
\item [$i)$] $\gcd(q-1,m)=1$, 
\item [$ii)$] $p$ does not divide $m$,
\item [$iii)$] the polynomial
\begin{equation*}
 Q(X)=   X^{\sigma^2+1}-X^{\sigma+1}
-X^\sigma+X \in \F_{q}[X]
\end{equation*}
has no roots in $\F_{q^m} \setminus \F_q$,
\end{itemize}
then $\mathcal{U}_\sigma$ is scattered.
\end{theorem}
\begin{corollary}
    Let $q=2$, $s=1$, and $\gcd(m,3)=1$, $m\geq 5$ odd. The subspace $\cU_\sigma$ defined in \eqref{stsubspace} is scattered.
\end{corollary}
\begin{proof}
The hypotheses $i)$ and $ii)$ of Theorem  \ref{MainTheorem} are trivially satisfied and $iii)$ asks that the polynomial $X^5+X^3+X^2+X=X(X+1)(X^3+X^2+1)$  has no roots in $\F_{2^m}\setminus\F_2$. This is equivalent to say that the polynomial $X^3+X^2+1$ has no roots in $\F_{2^m}$. Since the latter polynomial has degree $3$, has no roots in $\F_2$ and $\gcd(m,3)=1$, it also has no roots in $\F_{2^m}$. This concludes the proof.
\end{proof}
\begin{corollary}
    Let $q=3$, $s=1$ and $\gcd(m,12)=1$, $m\geq 5$. The subspace $\cU_\sigma$ defined in \eqref{stsubspace} is scattered.
\end{corollary}
\begin{proof}
The hypotheses $i)$ and $ii)$ of Theorem  \ref{MainTheorem} are trivially satisfied and $iii)$ asks that the polynomial $X^{10}-X^4-X^3+X=X(X+1)(X-2)(X^3+ 2X^2 + 2X + 2)(X^4+X^3+2X+1)$ has no roots $\F_{3^m}\setminus\F_3$, or equivalently, the polynomials $X^3+ 2X^2 + 2X + 2$ and $X^4+X^3+2X+1$, have no roots in $\F_{3^m}$. Since the latter polynomials are irreducible and have degrees $3$ and $4$, they have no roots in $\F_3$ and $\gcd(m,12)=1$, they also have no roots in $\F_{3^m}$. This concludes the proof.
\end{proof}
\begin{remark}
Along the lines of the corollaries above, once $q$ and $s$ are fixed, one could easily find other conditions on $m$, to guarantee $\cU_\sigma$ to be scattered.
\end{remark}

Note that each element of $\F_q$ is a root of the polynomial equation $Q(X)=0$; also, $Q(x)$ factorizes as follows:
\begin{equation}\label{eqfattorizzata}
Q(X)=Q_1(X)Q_2(X),
\end{equation}
where $Q_1(X)=X^{\sigma}-X$ and $Q_2(X)=X(X^{\sigma}-X)^{\sigma-1}-1$.
The polynomial $Q_2(X)$ has degree $q^{2s}-q^s+1$ with $\gcd(s,m)=1$. 

Taking into account the factorization expressed in (\ref{eqfattorizzata}), one sees that Condition $(iii)$ of Theorem \ref{MainTheorem} is equivalent to require that none of the roots of 
   $$X(X^{\sigma}-X)^{\sigma-1}-1 \in \F_q[X]$$
belong to $\fqm$. In different words it can be expressed as follows:
\begin{itemize}
 \item  [$iii)'$]  $\gcd(r,m)=1$ where $r$ is the degree of the splitting field  over $\F_q$ of the polynomial
   $$X(X^{\sigma}-X)^{\sigma-1}-1 \in \F_q[X]$$
\end{itemize}

 Recall that a polynomial $f(X)$ of degree $n$ with coefficients over a field $\F$  has splitting field $\mathbb{K}$ of degree at most $n!$ over $\F$. 
Then we get the following
\begin{corollary}
\label{cor:fattoriale} Let $\cU_\sigma$ be the $(m+2)$-dimensional subspace defined in \eqref{stsubspace}. If $\gcd(m,(q^{2s}-q^{s}+1)!)=1,$ then $\cU_\sigma$ is scattered.
\end{corollary}
\begin{proof}
   Since  $\gcd(m,(q^{2s}-q^{s}+1)!)=1$, $m$ is coprime with any integer  less than or equal to $(q^{2s}-q^{s}+1)!$. Hence, the hypotheses $i)$, $ii)$ of Theorem \ref{MainTheorem}  and $iii)'$ are trivially satisfied.
\end{proof}

Corollary \ref{cor:fattoriale} above provide us with a rough condition on the couples $(m,q)$ assuring that $\cU_\sigma$ is scattered.\\
In the remainder of this section, we will determine when the polynomial $Q(X)$ has no roots in $\fqm \setminus \fq$. For this purpose,  using the results and notations in \cite{SheekeyMcGuire2019}, we prove the following

\begin{proposition}\label{numberroots}
Let $m \geq 5$ be an odd integer, $\sigma: x \in \F_{q^m} \longrightarrow x^{q^s} \in \F_{q^m}$ with $\gcd(m,s)=1$.  The polynomial
\begin{equation*}
Q(X)=X^{\sigma^2+1}-X^{\sigma+1}-X^\sigma+X \in \F_{q}[X]
\end{equation*}
has  exactly the elements of $\F_q$ as roots in $\F_{q^m}$ if and only if the set 
$$\{\lambda \in\fqm\mid P_\gamma(\lambda)=0 \textnormal{ for some }\gamma\in\fq\}$$ has size at most $q$, where 
$P_\gamma(X)=X^{\sigma+1}-\gamma X + \gamma \in \F_q[X]$
with $\gamma \in \F_q$;   namely, if and only if $G_{m-1}(\gamma) \not = 0$ for any $\gamma \in \F^*_q$. 
\end{proposition}

\begin{proof}
   Suppose $G_{m-1}(\gamma) \not =0$ for any $\gamma \in \F^*_q$,  let $\overline{\mathbb{F}}_q$ be the algebraic closure of $\mathbb{F}_q$ and $\lambda \in \overline{\mathbb{F}}_q$ be a root of $Q(X)$. The expression
    \begin{equation*}
    \lambda^{\sigma^2+1}-\lambda^{\sigma+1}
-\lambda^\sigma+\lambda=0
\end{equation*}
can be rewritten as
\begin{equation*}
\lambda(\lambda+1)^{\sigma^2}-\lambda^\sigma(\lambda+1)=0.
\end{equation*}
Assuming $\lambda\neq0, -1$, this leads to
\begin{equation*}
(\lambda+1)^{\sigma^2-1}=\lambda^{\sigma-1},
\end{equation*}
namely
\begin{equation*}
\left (\frac{(\lambda+1)^{\sigma+1}}{\lambda}\right )^{\sigma-1}=1.
\end{equation*}
This is the case if and only if $\frac{(\lambda+1)^{\sigma+1}}{\lambda}=\gamma$ for some $\gamma\in\fq^*$. Hence, 
\begin{equation}\label{withgamma}
\lambda^{\sigma+1}+\lambda^\sigma+(1-\gamma)\lambda+1=0 \textnormal{\quad for } \gamma \in \F_q^*
\end{equation}
holds. Applying to \eqref{withgamma}  the substitution  $\lambda \mapsto \lambda-1$, one gets the equation
\begin{equation*}
    \lambda^{\sigma+1}-\gamma\lambda+\gamma=0 \quad \textnormal{ for } \gamma \in \F^*_{q}.
\end{equation*}
For any $\gamma \in \F_q$, define 
\begin{equation*}
    R_\gamma= \{ \lambda \in \F_{q^m} : \lambda^{\sigma+1}+\lambda^\sigma+(1-\gamma)\lambda+1=0 \} \, \textnormal{ and }\,   R'_\gamma= \{ \lambda \in \F_{q^m} : \lambda^{\sigma+1}-\gamma\lambda+\gamma=0\},
\end{equation*} and note that $0 \not \in R_\gamma$ for any $\gamma \in \F_q$ and so $1 \not \in R'_\gamma$ for any $\gamma \in \F_q$. Also, $-1 \in R_{0}$.

Then
\begin{equation}\label{containment}
 | \{ \lambda \in \F_{q^m} \colon  Q(\lambda)=0 \} |  \leq 
\biggl \vert
 \bigcup_{\gamma \in \F_q} R_\gamma
 \biggr  \vert +1 =  \biggl \vert  \bigcup_{\gamma \in \F_q} R'_\gamma \biggr \vert +1= \sum_{\gamma \in \F^*_q} \vert R'_\gamma \vert + 2.
\end{equation}
where the last equality follows because the sets $R'_{\gamma}$
's, with $\gamma \in \F_{q}$, are pairwise disjoint.

Now, let $L_\gamma(X)=X^{\sigma^2}-\gamma X^\sigma + \gamma X$ be the linearized polynomial associated with $P_\gamma(X)$ and consider $A_\gamma=A_{L_\gamma}$ the $m$-th power of companion matrix 
\begin{equation*}
C_\gamma=\begin{pmatrix}
0 & -\gamma \\
1 & \gamma
\end{pmatrix}
\end{equation*}
of $L_\gamma(X)$. We shall distinguish two cases depending on the field characteristic:
\begin{itemize}
\item [-] Case $q$ even. By hypothesis and by Equation \eqref{eq:description_Tr(A)} (where $A_L$ is now $A_\gamma$), one gets $\mathrm{Tr}(A_\gamma)=\gamma^mG_{m-1}(\gamma)\not =0$ for any $\gamma \in \F^*_q$, then the algebraic multiplicity of the eigenvalues of $A_\gamma$ is $1$. Hence, as a consequence of Theorem \ref{eigen}, the size of any $R'_\gamma$ ($\gamma \in \F^*_q$) is either $2$ or $0$, depending on whether $T(\gamma)=\mathrm{Tr}_{q/2}(\Lambda(\gamma))$ is $0$ or $1$, respectively, with
\begin{equation*}
\Lambda(\gamma)=\frac{\det(A_\gamma)}{\mathrm{Tr}(A_\gamma)^2}=\frac{1}{\gamma^m G^2_{m-1}(\gamma)}.
\end{equation*}
Hence
\begin{equation*}
    \sum_{\gamma \in \F^*_q} |R'_\gamma|+2 =  \sum_{\substack{\gamma \in \F^*_q\\ T(\gamma)=0}}|R'_\gamma|+ \sum_{\substack{\gamma \in \F^*_q\\ T(\gamma)=1}}|R'_\gamma|+2 \leq 2(2^{e-1}-1)+2=q.
    \end{equation*} 
    Since $Q(X)$ in \eqref{polynomial} has at least the elements of $\F_q$ as roots, the claim follows.
\item [-]  Case $q$ odd. By \cite{SheekeyMcGuire2019}, $G_{m-1}(\gamma) \not = 0$ for any $\gamma \in \F^*_q$, if and only if the size of any $R'_\gamma$ is 0,1, or 2, depending on whether the discriminant $\Delta_\gamma$ of the characteristic polynomial of $A_\gamma$ is a non-square, zero or a nonzero square of $\F_q$, respectively. By \cite[Proposition 4]{SheekeyMcGuire2019} and since $G_{k}(\gamma) \in \F_{q}$ for $k\in\{0,\ldots,m\}$, one gets
\begin{equation}
    \begin{aligned}
\Delta_\gamma=&\mathrm{Tr}(A_\gamma)^2-4\det(A_\gamma)\\
&=\gamma^{2m}\Bigl [ (2G_{m} (\gamma)+G_{m-1}(\gamma))^2-4 \Bigl (G^2_m(\gamma)+G_{m-1}(\gamma)G_{m}(\gamma)+\frac{G_{m-1}^2(\gamma)}{\gamma} \Bigr )  \Bigr]\\
&=\gamma^{2m-1}(\gamma-4)G^2_{m-1}(\gamma)
\end{aligned}
    \end{equation}
Since $m=2r+1$ for some integer $r \geq 2$, 
\begin{equation}
\Delta_\gamma=\gamma^{4r+1}(\gamma-4)G_{m-1}^2(\gamma)=\gamma(\gamma-4)(\gamma^{2r}G_{m-1}(\gamma))^2.
\end{equation}
So, by assumption, the estimate in \eqref{containment} becomes
\begin{equation}\label{Rodd}
\begin{aligned}
    \vert \{\lambda \in \F_{q^m} \colon Q(\lambda)=0 \}\vert &\leq \sum_{\gamma \in \F^*_q} |R'_\gamma|+2 =  \sum_{\substack{\gamma \in \F^*_q\\ \Delta_\gamma=0}}|R'_\gamma|+ \sum_{\substack{\gamma \in \F^*_q\\ \Delta_\gamma \not \in \square}}|R'_\gamma|+ \sum_{\substack{\gamma \in \F^*_q\\ \Delta_\gamma \in \square}}|R'_\gamma|+2\\
    &=|R'_4|+ \sum_{\substack{\gamma \in \F^*_q\\ \Delta_\gamma \not \in \square}}|R'_\gamma|+ \sum_{\substack{\gamma \in \F^*_q\\ \Delta_\gamma \in \square}}|R'_\gamma|+2.
    \end{aligned}
    \end{equation}
 Since $\Delta_\gamma \in \square$ if and only if $\gamma(\gamma-4) \in \square$, in order to estimate \eqref{Rodd}, we have to compute the size of the set
\begin{equation}\label{square}
    \{(\gamma, \gamma(\gamma-4)) \in \F^2_q \setminus \{(0,0)\}   \colon \gamma(\gamma-4)) \in \square\}.
\end{equation}
Let  
\begin{equation*}
    \Gamma:\xi^2-4\xi=\eta^2
\end{equation*}
be a 
hyperbola in the  affine plane $\mathrm{AG}(2,q)$ whose points have coordinates $(\xi,\eta)$. Since for any choice of $\gamma \in \F_q$ with $\gamma \not =0,4$, there exists two of points of $\Gamma$ with coordinates $(\gamma,\pm \delta)$, the size of the set in \eqref{square} is $(q-3)/2$.
So, Formula \eqref{Rodd} becomes

\begin{equation*}
\begin{aligned}
  &|R'_4|+ \sum_{\substack{\gamma \in \F^*_q\\ \Delta(\gamma) \not \in \square}}|R'_\gamma|+ \sum_{\substack{\gamma \in \F^*_q\\ \Delta(\gamma) \in \square}}|R'_\gamma|+2\\
    & =1+ 2 \cdot\frac{q-3}{2} +2= q.
    \end{aligned}
    \end{equation*}
Since $Q(X)$ has at least the elements of $\F_q$ as roots, the claim follows.
\end{itemize}

Vice versa, suppose that there exist $\overline{\gamma} \in \F^*_q$ such that $G_{m-1}(\overline{\gamma})=0$. Then, by  \cite[Theorem 8]{SheekeyMcGuire2019}, the projective polynomial 
\begin{equation*}
    X^{\sigma+1}-\overline{\gamma}X+\overline{\gamma}=0 
\end{equation*}
and so \begin{equation*}
X^{\sigma+1}+X^\sigma+(1-\overline{\gamma})X+1=0
\end{equation*}
has $q+1$ solutions, whence the polynomial $Q(X)$ has at least a root not beloging to $\F_q$.
\end{proof}

\begin{remark}
Note that in even characteristic, one obtains an explicit formula 
for $G_{m-1}(\gamma)$ with $\gamma \in \F^*_q$. Indeed,
$$G_{m-1}(\gamma)=\frac{\mathrm{Tr}(A_\gamma)}{\gamma^m},$$
and  by \cite[Corollary 4.1]{patil2021trace}, since  $A_\gamma=C_\gamma^m$, one has that 
$$G_{m-1}(\gamma)=\sum^{\lfloor m/2 \rfloor}_{j=0}\frac{m(m-j-1)!}{j!(m-2j)!}\frac1{\gamma^j}.$$

Instead, in odd characteristic, it is rather demanding in general to figure out for which couples $(q,m)$, $G_{m-1}(\gamma)$ is not zero for any $\gamma \in \F^*_{q}$, because of the recursive Formula (\ref{eq:recursion}) defining $G_{m-1}(\gamma)$. 
\end{remark}

However, taking into account Theorem \ref{MainTheorem} and Proposition \ref{numberroots}, we have 
\begin{corollary}\label{gscattered} Let $q = p^e$, $m \geq 5$, $\gcd(q-1,m)=1$ and $p \nmid m$. Consider the $(m + 2)$-dimensional
subspace
$$\mathcal{U}_{\sigma} = \{(x, x^\sigma + a, x^{\sigma^2}
+ b) : x  \in \F_{q^m}, a, b \in \F_q\}$$
of $\F_{q^m}^3$ , where $\sigma : x \in \F_{q^m} \longrightarrow x^{q^s}
\in \F_{q^m}$ , $1 \leq  s \leq m -1$ and $\gcd(s, m) = 1$. If $G_{m-1}(\gamma)\not =0$ for any $\gamma \in \F^*_{q^m}$, then $\mathcal{U}_\sigma$ is scattered.
\end{corollary}

In the following, we will restrict to $m \in \{5,7\},$ in which case we will get arithmetic conditions on $q$ and $m$ such that $\mathcal{U}_\sigma$ is scattered.


\subsection{Case \texorpdfstring{$m=5$}{}} \label{subsec:Case5}
Note that for $m=5$, $m+2$ is equal to $\lfloor{\frac{3m}{2}}\rfloor$ and by \cite{blokhuis2000scattered} this is the maximal possible dimension of a scattered subspace in $\F^3_{q^5}$. 

Let 
\begin{equation*}
X^{\sigma+1}-\gamma X + \gamma \in \F_q[X]
\end{equation*}
with $\gamma \in \F^*_{q}$,
By the table in \cite[Subsection 4.4]{SheekeyMcGuire2019}, we get
\begin{equation*}
    G_4(\gamma)=\frac{\gamma^2-3\gamma+1}{\gamma^2}.
\end{equation*}

\begin{itemize}
    \item [-] \textit{Case $q$ even.} By Proposition \ref{numberroots}, Condition $iii)$ of Theorem \ref{MainTheorem} is satisfied if and only if  $G_{4}(\gamma)\not= 0$ for any $\gamma \in \F^*_{q}$ and hence $q=2^e$ with $e$ an odd integer. Since  $2^e \equiv 2\cdot 4^h\equiv 2 \cdot (\pm 1) \equiv \pm 2\pmod{5}$, $h \in \mathbb{N}$, Condition $i)$ and $ii)$ in Theorem \ref{MainTheorem} hold as well. Hence $\cU_{\sigma} \subset \F^3_{q^5}$ is scattered.
    \item [-] \textit{Case $q$ odd.} $G_4(\gamma) \not = 0$ for any $\gamma \in \F^*_q$ if and only if $5$ is not a square, namely if and only if $q \equiv 2,3 \pmod 5.$ 
Conditions $i)$,$ii)$ and $iii)$ in Theorem \ref{MainTheorem} are satisfied and so $\cU_\sigma$ is scattered.\\

Summing up we get the following theorem.

\begin{theorem}\label{thm:m5}
Let $q=p^e$ and
\begin{equation*}
\cU_\sigma=\{(x,\xqs+a,\xxs +b):x \in\F_{q^5},a,b\in\fq\} \subset \F_{q^5}^3,
\end{equation*}
where $\sigma: x \in \F_{q^5} \longrightarrow x^{q^s} \in \F_{q^5}$, $s \in \{1,2,3,4\}$. Then,
\begin{itemize}
\item [-] for $p=2$, $\cU_\sigma$ is maximum scattered if $e$ is an odd positive integer.
\item [-] for $p$ odd, $\cU_\sigma$ is maximum scattered if $q\equiv 2,3 \pmod 5$.
\end{itemize}
\end{theorem}

\end{itemize}

\subsection{Case \texorpdfstring{$m=7$}{}}\label{subsec:Case7}
Now suppose $m=7$. Using (4) 
of \cite[Section 4]{SheekeyMcGuire2019}, it follows that for any $\gamma\in \F^*_q$, 
\begin{equation*}
  G_6(\gamma)=\frac{\gamma^3-5\gamma^2+6\gamma-1}{\gamma^3}
\end{equation*}
We split our discussion depending on $\mathrm{char}(\F_q)=p$:
\begin{itemize}
    \item [-] if $p=2$, $G_6(\gamma)=\frac{\gamma^3+\gamma^2+1}{\gamma^3}$ and since its numerator is an irreducible polynomial over $\mathbb{F}_2$, $G_6(\gamma)\neq 0$ for $q=2^e$ and $3\nmid e$. So, by Proposition \ref{numberroots}, $\cU_\sigma$ is scattered.

    \item [-] if $p$ is odd, 
\begin{itemize}
    \item [$a)$] for $p=3$, for any $\gamma \in \F^*_q$, $G_6(\gamma)=\frac{\gamma^3+\gamma^2+1}{\gamma^3}$ and, as before, since its numerator is an irreducible polynomial over $\F_3$, $G_6(\gamma) \not = 0$ for any $q=3^e$ and $3\nmid e$.
    \item [$b)$] for $p=5$, for any $\gamma \in \F^*_q$, $G_6(\gamma)=\frac{\gamma^3+\gamma+4}{\gamma^3}$ and, as before, since its numerator is an irreducible polynomial over $\F_5$, $G_6(\gamma) \not = 0$ for any $q=5^e$ and $3\nmid e$.
    \item [$c)$] for $p=7$, $G_6(\gamma)=\frac{(\gamma+3)^3}{\gamma^3}$, so by Proposition \ref{numberroots}, Condition $iii)$ in \ref{MainTheorem} is not satisfied.
    \item [$d)$] for $p>7$, applying the transformation $\gamma \mapsto \gamma+\frac{5}{3}$ to the equation $\gamma^3-5\gamma^2+6\gamma-1=0$, one gets 
\begin{equation*}
\gamma^3-\frac{7}{3}\gamma-\frac{7}{27}=0.
\end{equation*}
By \cite[Theorem 3]{Dickson1906}, this is irreducible over $\F_q$ and so  it has no roots in $\fq$ if and only if $\frac{7}{18} \cdot (\frac1{3}+\sqrt{-3})$ is not a cube in $\F_q(\sqrt{-3})$.
\end{itemize}
\end{itemize}
Finally for sake of completeness, we collect all the results of this Subsection in the following theorem.
\begin{theorem}\label{thm:m7}
    Let $q=p^e$ and
\begin{equation*}
\cU_\sigma=\{(x,\xqs+a,\xxs +b):x \in\F_{q^7},a,b\in\fq\} \subset \mathbb{F}^3_{q^7},
\end{equation*}
where $\sigma: x \in \F_{q^7} \longrightarrow x^{q^s} \in \F_{q^7}$, $s \in \{1,\ldots,6\}$. Then,
\begin{itemize}
\item [-] for $p=2,3,5$, $\mathcal{U}_\sigma$ is scattered if $3 \nmid e$.
\item [-] for $p>7$, $\cU_\sigma$ is scattered if $\frac{7}{18} \cdot (\frac1{3}+\sqrt{-3})$ is not a cube in $\F_q(\sqrt{-3})$.
\end{itemize}

\end{theorem}

\section[Equivalence issue for U]{Equivalence issue for $\cU_{\sigma}$} \label{sec:equivalence}
Let $\mathcal{U}_{\sigma}$ be the $\fq$-subspace defined in \eqref{stsubspace} with $\sigma: x \in \fqm \rightarrow x^{q^s} \in \fqm$ with $ 1 \leq s<m$ and $\gcd(s,m)=1$. Throughout this section, algebraic conditions involving the integer $s$ will be analyzed; for this reason we will denote $\mathcal{U}_\sigma$ by $\mathcal{U}_s$.\\
Given $1\leq s,t<m$, $\gcd(s,m)=\gcd(t,m)=1$, we are going to investigate whether $\cU_s$ and $\cU_t$ are $\mathrm{\Gamma L}(3,q^m)$-equivalent. Note that this amounts to investigate their $\GL(3,q^m)$-equivalence. Note that this reduces to investigate their $\GL(3,q^m)$-equivalence.
Indeed, if $\cU_s$ and $\cU_t$ are $\GaL(3,q^m)$-equivalent, there exists $A\in \GL(3,q^m)$ and $\zeta \in \Aut(\F_{q^m})$ such that
\begin{equation*}
\left (x^\zeta,(x^\sigma)^\zeta +a^\zeta,(x^{\sigma^2})^\zeta+b^\zeta \right )    A=(y,y^\tau+a',y^{\tau^2}+b').
\end{equation*}
and hence
\begin{equation*}
\left (x^\zeta,(x^\zeta)^\sigma +a^\zeta,(x^\zeta)^{\sigma^2}+b^\zeta \right )    A=(y,y^\tau+a',y^{\tau^2}+b').
\end{equation*}
By replacing $\bar{x}=x^\zeta, \bar{a}=a^\zeta$ and $\bar{b}=b^\zeta$, we get that $\cU_s$ is $\GL(2,q^m)$-equivalent to $\cU_t$.
Assume that $\cU_s$ and $\cU_t$ are $\GL(3,q^m)$-equivalent, and let $A=(a_{ij})$ be the invertible matrix such that $\cU_s \cdot A =\{vA:v \in \cU_s\}=\cU_t$. Then, for any $(x,a,b)\in \fqm\times\fq\times\fq$, there exists $(y,a',b')\in\fqm\times\fq\times\fq$ satisfying 

\begin{equation}\label{equivissue}
(x,x^\sigma +a,x^{\sigma^2}+b)    A=(y,y^\tau+a',y^{\tau^2}+b').
\end{equation}
This holds in particular for the vectors corresponding to $(x,0,0)$. In this case we obtain
\begin{equation*}
\begin{split}
&a_{12}x+a_{22}x^{q^s}+a_{32}x^{q^{2s}}-(a^{q^t}_{11}x^{q^t}+a_{21}^{q^t}x^{q^{s+t}}+a_{31}^{q^t} x^{q^{2s+t}})=a' \textnormal{ and }\\
&a_{13}x+a_{23}x^{q^s}+a_{33}x^{q^{2s}}-(a_{11}^{q^{2t}}x^{q^{2t}}+a^{q^{2t}}_{21}x^{q^{s+2t}}+a^{q^{2t}}_{31}x^{q^{2(s+t)}})=b'.
\end{split}
\end{equation*}
Since $a',b'\in\fq$, raising this expressions to the power $q^i$, for any  $i \in \{ 1,\ldots, m-1\}$ leads to
\begin{equation}\label{equivsystem}
\begin{cases}
&a_{12}x+a_{22}x^{q^s}+a_{32}x^{q^{2s}}-(a^{q^t}_{11}x^{q^t}+a_{21}^{q^t}x^{q^{s+t}}+a_{31}^{q^t} x^{q^{2s+t}})\\
&=
a_{12}^{q^i}x^{q^i}+a_{22}^{q^i}x^{q^{s+i}}+a_{32}^{q^i}x^{q^{2s+i}}-(a^{q^{t+i}}_{11}x^{q^{t+i}}+a_{21}^{q^{t+i}}x^{q^{s+t+i}}+a_{31}^{q^{t+i}} x^{q^{2s+t+i}})
\\
&a_{13}x+a_{23}x^{q^s}+a_{33}x^{q^{2s}}-(a_{11}^{q^{2t}}x^{q^{2t}}+a^{q^{2t}}_{21}x^{q^{s+2t}}+a^{q^{2t}}_{31}x^{q^{2(s+t)}})\\&=
a_{13}^{q^i}x^{q^i}+a_{23}^{q^i}x^{q^{s+i}}+a_{33}^{q^i}x^{q^{2s+i}}-(a_{11}^{q^{2t+i}}x^{q^{2t+i}}+a^{q^{2t+i}}_{21}x^{q^{s+2t+i}}+a^{q^{2t+i}}_{31}x^{q^{2(s+t)+i}})\\
& i=1,\ldots,m-1
\end{cases}
\end{equation}
Assume that the elements of the set of $q$-powers $S=\{0,s,2s,t,s+t,2s+t\}$ on the left hand side of the first equation are pairwise distinct modulo $m$ (note that this is possible for $m \geq 7$). Since for any $q$-power in $S$ there exists $i \in \{ 1, \ldots, m-1\}$ such that this $q$-power does not appear in $i+S$ $\pmod m$. This implies $a_{12}=a_{22}=a_{32}=a_{11}=a_{21}=a_{31}=0$, and hence $\cU_s$ and $\cU_t$ are not equivalent.

The integers in $S=\{0,s,2s,t,s+t,2s+t\}$ fail to be pairwise distinct modulo $m$ if and only if $s \equiv t$, $2s \equiv t$, $s+t \equiv 0$ or $2s+t \equiv 0 \pmod m$.
\begin{itemize}
\item [-] Case $s+t \equiv 0 \pmod m $. Up to interchanging $s$ with $t$, we may suppose that $s > m/2$. 
Then $2m > 2s=m+\ell$ for some integer $\ell < m$. Observe that $s-\ell=m-s>0$, $m-\ell=2(m-s)>0$ and
\begin{equation*}
\mathcal{U}_{s}=\{((x^{q^\ell})^{q^{m-\ell}},(x^{q^{\ell}})^{q^{s-\ell}}+a,x^{q^{\ell}}+b) : x \in \fqm ,a,b \in \fq\}.
\end{equation*}
Putting $z=x^{q^\ell}$ then
\begin{equation*}
\mathcal{U}_{s}=\{(z^{q^{2(m-s)}},z^{q^{m-s}} +a,z+b) : z \in \fqm ,a,b  \in \fq\}=\{(y^{q^{2t}}+b,y^{q^t}+a,y) : y \in \fqm ,a,b \in \fq\}
\end{equation*}
with $t=m-s$ and $(m,t)=1$. Clearly, $\mathcal{U}_s$ is $\GL(3,q^m)$-equivalent to the subspace $\mathcal{U}_{t}$.
\item [-] Case $2s\equiv t \pmod m$. The elements of the set of $q$-powers $S=\{0,s,t,s+t,2t\}$ on the left hand side of the first equation are pairwise distinct modulo $m$. Since for any $q$-power in $S$ there exists $i \in \{ 1, \ldots, m-1\}$ such that this $q$-power does not appear in $i+S$ $\pmod m$. Again, by polynomial identity, we get $a_{12}=a_{22}=a_{21}=a_{31}=a_{13}=a_{23}=0$, which since $A\in \GL(3,q^m)$, can not be the case.

\item [-]
Case $2s+t \equiv 0 \pmod m$. Let $s'=2s, t'=t$, then $s'+t' \equiv 0 \pmod m$. By the first case, $\cU_{s'}$ is equivalent to $\cU_{t'}$. Now, if $\cU_{s}$ is equivalent to $\cU_{t}$ then it is equivalent to $\cU_{s'}$ and this is excluded by the previous case.
\end{itemize}
We have then proved the following result.
\begin{theorem}
Let $\mathcal{U}_s$ and $\mathcal{U}_t$ be the subspaces as defined in \eqref{stsubspace}, with  $1 \leq s,t < m$ and $\gcd(s,m)=\gcd(t,m)=1$. They are $\GaL$-equivalent if and only if $t \in \{s,m-s\}$.
\end{theorem}
In particular, if $s=t$, System \eqref{equivsystem} leads to $a_{12}=a_{32}=a_{31}=a_{13}=a_{21}=0$, $a_{22}=a_{11}^{q^s}$ and $a_{33}=a_{11}^{q^{2s}}$. This amounts to saying
\begin{equation*}
   A= \begin{pmatrix}
   \alpha & 0 & 0 \\
   0 & \alpha^{q^s} & 0 \\
   0 & 0 & \alpha^{q^{2s}}
    \end{pmatrix},
\end{equation*}
with $\alpha=a_{11} \in \F^*_{q^m}$. Since this must hold true for $x=b=0$ in System \eqref{equivissue}, it follow that $\alpha \in \fq^*$ .
We have proved the following result
\begin{theorem}
The stabilizer of $\mathcal{U}_\sigma$ in $\GaL(3,q^m)$ is isomorphic to $\fq \rtimes \Aut(\fqm)$.
\end{theorem}

\section[Linear sets associated with the cutting blocking sets U]{Linear sets associated with the cutting blocking sets $\cU_{\sigma}$}
\label{sec:characters}

Let $\PG(\mathbb{F}_{q^m}^k, \fqm)=\PG(k-1,q^m)$. If $\cU$ is an $\fq$-subspace of dimension  $n$, then the set of points 
\begin{equation*}
L_\cU= \{ \langle u \rangle_{\mathbb F_{q^m}}\, \colon \, u \in \cU \setminus \{{\bf 0}\}  \} \subset \PG(k-1,q^m)
\end{equation*}
is called an $\mathbb{F}_q$-\textit{linear set} of $\PG(k-1,q^m)$ of rank $n$.

Let $\Lambda = \PG(T,q^m)$ be a projective subspace of $\PG(k-1,q^m)$ with underlying vector space $T \subseteq \F_{q^m}^k$. Then $L_{\cU \cap T}=L_\cU \cap \Lambda,$ and the \textit{weight} of $\Lambda$ in $L_\cU$ is defined as
\begin{equation*}
w_{L_\cU}
(\Lambda) = \dim_{\mathbb{F}_q}(\cU \cap T). 
\end{equation*}
If $\text{dim}_{\mathbb{F}_q}(\cU \cap T)=i$, one shall say that $\Lambda$ has {\it weight} $i$ in $L_\cU$.

Then $L_{\cU \cap T}=L_\cU \cap \Lambda,$ and the \textit{weight}  $w_{L_\cU}(\Lambda)$ of $\Lambda$ in $L_\cU$ is defined as the weight of $T$ in $\cU$.
If  $w_{L_\cU}(\Lambda)=i$, one shall say that $\Lambda$ has {\it weight} $i$ in $L_\cU$.

\subsection{Characters of $L_{\cU_{\sigma}}$ with respect to the lines of $PG(2,q^m)$}

In \cite[Theorem 4.2]{blokhuis2000scattered}, the authors proved that if $\cU$ is  an $\F_q$-linear scattered subspace of rank $km/2$, $km$ even, in $\F_{q^m}^k$, then the linear set $L_\cU$ has two intersection characters in $\PG(k-1, q^m)$ with respect to the hyperplanes. 
As seen in Section \ref{sec:construction} , the $(m+2)$-dimensional subspace $\cU_\sigma$ of $\F_{q^m}^k$ contains a 2-scattered subspace $\mathcal{W}_\sigma$ of dimension $m$ (cf. \cite[Lemma 2.2]{CMPZ2021-generalising}). We will show that this configuration gives rise to a linear set with a wider range of intersection characters with respect to hyperplanes. 
More precisely, let $L_\cU\subseteq \PG(k-1,q^m)$ be a linear set such that $ \dim_{\F_q}\cU =k+m-1$, $k \geq 3$. If the subspace $\mathcal{U}$ contains  
 an $m$-dimensional $(k-1)$-scattered subspace $\mathcal{W}$, then
for each hyperplane $\Lambda$ we get
\begin{equation}\label{spec.characters}
k-1\leq  w_{L_\cU}(\Lambda)\leq 2(k-1),
\end{equation}
i.e. $L_\cU$ has at most $k$ characters with respect to the hyperplanes of $\PG(k-1,q^m)$. 
The first inequality  \eqref{spec.characters} is an easy consequence of the Grassmann formula. For the second one,
\begin{equation*}
\begin{aligned}
 w_{L_\cU}(\Lambda)&=\dim_{\F_q}(H \cap \mathcal{U})=\dim_{\F_q}(  H \cap \mathcal{U} \cap \cW)+\dim_{\F_q}( (H \cap \cU) + \cW) - \dim_{\F_q}\cW \\
& \leq \dim_{\F_q}(H \cap \cW) + \dim_{\F_q} \cU -\dim_{\F_q} \cW \leq k-1 + (k + m  -1) - m= 2(k-1)
\end{aligned}
\end{equation*}
where the last inequality follows because of $(k-1)$-scatteredness of $\cW$.\\

In order to compute the characters of $L_{\mathcal{U}_\sigma}$ with respect to the lines of $\PG(2,q^m)$, some identities of the admissible characters of intersection with respect to the hyperplanes of any subset of a projective space must be recalled. These identities are known as the \textit{standard equations}, and can be obtained by a simple double counting argument, see e.g. \cite[Lemma 6]{Honold2018}. 
\begin{lemma}\label{standardequationsLemma}
Let $\mathcal{S}$ be a subset of the projective space $\PG(v-1,q)$, with $|\mathcal{S}|=n$. For any $i\in\{0,1,\dots,n\}$. Define $a_i$ as the number of hyperplanes meeting $\mathcal{S}$ in precisely $i$ points. Then the following identies hold:
$$\sum_{i=0}^na_i=\frac{q^v-1}{q-1};$$
$$\sum_{i=1}^nia_i=n\left(\frac{q^{v-1}-1}{q-1}\right);$$
$$\sum_{i=2}^n\binom{i}{2}a_i=\binom{n}{2}\frac{q^{v-2}-1}{q-1}.$$
\end{lemma}

\medskip


\begin{theorem} \label{th:characters}
Let $L_\cU\subseteq \PG(2,q^m)$ be a linear set such that $\cU$ is has rank $m+2$, $m\geq 5$. If $\mathcal{U}$ contains an $m$-dimensional $2$-scattered subspace $\mathcal{W}$, then the linear set $L_\cU$ has exactly three characters with respect to the lines. More precisely, there are exactly $A_2$, $A_3$ and $A_4$ lines whose intersection with $L_\cU$ has size $q+1$, $q^2+q+1$ and $q^3+q^2+q+1$ respectively, where $A_2$, $A_3$ and $A_4$ are as in \eqref{A2}, \eqref{A3} and \eqref{A4}.
\end{theorem}

\begin{proof}
Let  $A_i$ be the number of lines of $\PG(2,q^m)$ with weight $i$, $i \in \{ 0,1, \ldots, m+2\}$. Since 
$ 2 \leq w_{L_\cU}(\ell) \leq 4$ for any line $\ell$ of $\PG(2,q^m)$, one has
that  $A_i=0$ for any $i \not = 2,3,4$. We show that $A_2,A_3,A_4\ne 0$. By
Lemma \ref{standardequationsLemma}, since  $A_i=a_{\frac{q^i-1}{q-1}}$ and $|L_\cU|=\frac{q^{m+2}-1}{q-1}$, one gets
\begin{equation}\label{linearsystem}
\begin{cases}
&A_{2}+A_{3}+A_{4}=q^{2m}+q^m+1 \\
&(q^{2}-1)A_{2}+(q^{3}-1)A_{3}+(q^{4}-1)A_{4}=(q^{m+2}-1)(q^m+1) \\
&(q^{2}-1)(q-1)A_{2}+(q^{3}-1)(q^{2}-1)A_{3}+(q^{4}-1)(q^{3}-1)A_{4}=(q^{m+2}-1)(q^{m+1}-1)
\end{cases}
\end{equation}
where in the last equation we used the following identity

\begin{equation*}
\binom{\frac{q^{\ell}-1}{q-1}}{2}=\frac{q}{2}\frac{q^\ell-1}{q-1}\frac{q^{\ell-1}-1}{q-1},   \quad \quad \quad \ell \geq 2.
\end{equation*}
Then \eqref{linearsystem} is a linear system of three equations in the unknowns $A_{2},A_{3},A_{4}$ with associated matrix
\begin{equation*}
\begin{pmatrix}
1 & 1 & 1 \\
(q^{2}-1) & (q^{3}-1)& (q^{4}-1)\\
(q^{2}-1)(q-1) & (q^{3}-1)(q^{2}-1) & (q^{4}-1)(q^{3}-1)
\end{pmatrix}
\end{equation*}
whose determinant is $ q^{6} (q+1)(q-1)^3$, clearly not zero for $q \geq 2$.
So, solving this linear system, one gets:
\begin{equation}\label{A2}
\begin{aligned}
    A_{2}&=\frac{1}{q^{4}(q+1)(q-1)^2}(q^7 + q^{2 m} - q^{1 + m}- q^{5 + m}- q^{6 + m} + q^{
   7 + m} + q^{2 + 2 m} - q^{5 + 2 m} - q^{6 + 2 m} + q^{7 + 2 m})\\
   &=\frac{1}{q^{4}(q+1)(q-1)^2}(q^7+ q^{m+1}(q^{6}-q^5-q^4-1) +  q^{2 m}(q^{7}-q^6-q^5+q^2+1))>0
    \end{aligned}
\end{equation}
\begin{equation}\label{A3}
\begin{aligned}
  A_{3}=&\frac{1}{(q-1)^2q^4} (q^{m-1}-1)(q^5 - q^m - q^{3 + m} + q^{4 + m})=\\
  &\frac{1}{(q-1)^2q^4} (q^{m-1}-1)(q^5 + q^m(q^4-q^3-1))>0
  \end{aligned}
\end{equation}
\begin{equation}\label{A4}
      A_{4}=\frac{1}{(q+1)(q-1)^2}(q^{m-1}-1) (q^{m-4}-1)>0
\end{equation}
and, hence there are exactly three characters of $L_{\cU}$ with respect to the lines.
\end{proof}


\subsection{\bf Some consequences of previous results in terms of minimal codes}
As a byproduct of Theorem \ref{th:correspondence} and of Theorem \ref{th:characters} we obtain the exact weight distribution of codes corresponding to $\mathcal{U}$.
\begin{theorem}
Let $\mathcal{U}$ be a cutting blocking set of rank $m+2$, $m \geq 5$ odd and suppose that $\mathcal{U}$ contains an $m$-dimensional $2$-scattered subspace $\mathcal{W}$. Then the $[m+2,3]_{q^m/q}$  code $\cC\in\Psi(\left[  \mathcal{U}\right] )$ has minimum distance $d=m-2$. 
More precisely, there are precisely $A_2$, $A_3$ and $A_4$ codewords of rank weight $m$, $m-1$ and $m-2$ respectively, where $A_2$, $A_3$ and $A_4$ are as in \eqref{A2}, \eqref{A3} and \eqref{A4}.
\end{theorem}
\noindent Furthermore, the codes obtained from the scattered subspaces $\cU_\sigma$ have generator matrix of the shape
\begin{equation}\label{generator-matrix}
\begin{pmatrix}
 a_1 & a_2 & \ldots & a_m & 0 & 0 \\
  a^{q^s}_1 & a^{q^s}_2 & \ldots & a^{q^s}_m & 1 & 0 \\
  a_1^{q^{2s}} & a^{q^{2s}}_2 & \ldots & a_m^{q^{ms}} & 0 & 1
\end{pmatrix},
\end{equation}
where $a_1,\ldots, a_m$ is an $\F_q$-basis of $\F_{q^m}$, $1 \leq s \leq n-1$ and $\gcd(s,n)=1$. By checking the Singleton-like bound for rank metric codes in \eqref{singleton-bound}, one can verify that $d=m-2$ is the maximum possible value for the minimum rank of a $[m + 2, 3]_{q^m/q}$ rank metric code for every $m \geq 3$. The scattered subspace $\mathcal{U}_\sigma$ gives a rank metric code whose second generalized rank weight is $d_2=m + 1$  and this is the largest possible value, see for instance \cite[Theorem 3.3]{MarinoNeriTrombetti2023},  \cite[Proposition 2.5]{BMN}. Observe that a code with generator matrix obtaining by deleting the last two columns from \eqref{generator-matrix} is an $[m, 3,m - 2]_{q^m/q}$ Gabidulin code. We also point out that the shape of the matrix \eqref{generator-matrix}
reminds the one of triply extended Reed-Solomon codes over finite fields of even size which provide constructions for MDS codes, see \cite[Chapter 11]{bok:MW}.

\subsection{\bf Some consequences of previous results in terms of rank saturating systems}
The concept of {\it rank saturating system} is introduced for a rank-metric codes $\cC$ with given {\it rank covering radius} $\rho_{\rk}(\cC)$, see \cite[Definition 2.4]{bonini2022saturating}.  An $[n, k]_{q^m/q}$ system $\cU$ is \textit{rank} $\rho$-\textit{saturating} if the linear set $L_{\cU}$ is a $(\rho - 1)$-saturating set in $\PG(k - 1, q^m)$, i.e. if for any point $Q \in \PG(k-1,q^m)$, there exist $\rho$ points, say
$P_1,\ldots, P_{\rho} \in \mathcal{L}_\cU$ such that $Q \in \langle P_1,\ldots, P_{\rho} \rangle$ and $\rho$ is the smallest value with this property, \cite[Definition 3.2]{bonini2022saturating}.
As shown in \cite[Theorem 2.5]{bonini2022saturating}, an $[n, k]_{q^m/q}$ system  $\cU$ associated to a code $\cC$ is rank $\rho$-saturating if and only if $\rho_{\rk}(\cC^\perp)= \rho$.
Moreover, if  $\cU$ is an $[n,k]_{q^m/q}$ linear cutting blocking set, $\cU$ is
a rank $(k - 1)$-saturating $[n, k]_{
q^{m(k-1)}/q}$ system, see \cite[Theorem 4.6]{bonini2022saturating}.
Therefore, by our construction, for the couples $(m,q)$ such that $\cU_\sigma$ is a linear cutting blocking set, we have a family of rank $2$-saturating $[m+2,3]_{q^{2m}/q}$ systems with the property that the dual of the codes in $\Psi(\left[ \mathcal{U}_{\sigma}\right] )$ have rank covering radius $2$. 
So, by \cite[Corollary 4.7]{bonini2022saturating}, using the same arguments of \cite[Corollary 4.10]{bonini2022saturating}, one gets that, if $\cU_\sigma \leq \F^k_{q^m}$ is scattered, the minimal dimension $s_{q^{2m}/q}(k, \rho)$  of a rank $\rho$-saturating system in $\F^k_{q^m}$ is $s_{q^{2m}/q}(3,2)=m+2$.

\section{Conclusions} \label{sec:further results}
In this article we exhibit new examples of linear non-degenerate minimal $[m+2,3,m-2]_{q^m/q}$ rank-metric codes with $m\geq 5$ odd for infinite values of $q,$ providing conditions on $m$, such that their dual codes have rank covering radius equal to $2$. 

The keystone to achieve all this is the construction of a family of $(m+2)$-dimensional scattered linear sets of the projective plane $\PG(2,q^m)$, whose underlying vector spaces turn out to be linear cutting blocking sets of $\mathbb{F}^3_{q^m}$, and then elaborating on the link between these objects and rank-metric codes. 

Further computational research suggests the existence of $[m+2,3]_{q^m/q}$ minimal rank codes of this type also for other values of $m$ and $q$, for instance for $m=5$ and $q=4$. For this reason we think it would be interesting to push forward in the study of this family. Also it would be of a certain interest to look for other examples of maximum scattered linear subspaces in $\mathbb{F}_{q^m}^r$ when $rm$ is odd and $(r,m) \notin \{(3,3),(3,5)\}$. 


\section*{Acknowledgments}
The authors express their gratitude to the anonymous referees for their meticulous review.\\
This work was supported by the Italian National Group for Algebraic and Geometric Structures and their Applications (GNSAGA-- INdAM). The first author acknowledges the support by the Irish Research Council, grant n. GOIPD/2022/307. 

\bibliographystyle{abbrv}

\bibliography{scatteredplanebiblio}

\begin{thebibliography}{10}

\bibitem{Abhyankar}
S.~S. Abhyankar.
\newblock Projective polynomials.
\newblock {\em Proc. Am. Math. Soc.}, 125:1643--1650, 1997.

\bibitem{AlfaranoBorelloNeriRavagnani2022JCTA}
G.~N. Alfarano, M.~Borello, A.~Neri, and A.~Ravagnani.
\newblock Linear cutting blocking sets and minimal codes in the rank metric.
\newblock {\em J. Comb. Theory Ser. A}, 192:105658, 2022.

\bibitem{ball2000linear}
S.~Ball, A.~Blokhuis, and M.~Lavrauw.
\newblock Linear $(q+1)$-fold blocking sets in $\mathrm{PG}(2, q^4)$.
\newblock {\em Finite Fields their Appl.}, 6(4):294--301, 2000.

\bibitem{BartoliCsajbokMarinoTrombetti2021}
D.~Bartoli, B.~Csajb\'{o}k, G.~Marino, and R.~Trombetti.
\newblock Evasive subspaces.
\newblock {\em J. Combin. Des.}, 29(8):533--551, 2021.

\bibitem{bartoli2018maximum}
D.~Bartoli, M.~Giulietti, G.~Marino, and O.~Polverino.
\newblock Maximum scattered linear sets and complete caps in galois spaces.
\newblock {\em Combinatorica}, 38:255--278, 2018.

\bibitem{BMN}
D.~Bartoli, G.~Marino, and A.~Neri.
\newblock New \textnormal{MRD} codes from linear cutting blocking sets.
\newblock {\em Ann. di Mat. Pura ed Appl.}, 202(1):115--142, 2023.

\bibitem{Berger}
T.~Berger.
\newblock Isometries for rank distance and permutation group of gabidulin
  codes.
\newblock {\em IEEE Trans. Inf. Theory}, 49(11):3016--3019, 2003.

\bibitem{blokhuis2000scattered}
A.~Blokhuis and M.~Lavrauw.
\newblock Scattered spaces with respect to a spread in $\mathrm{PG}(n, q)$.
\newblock {\em Geom. Dedicata}, 81:231--243, 2000.

\bibitem{Bluher2004}
A.~W. Bluher.
\newblock On {$x^{q+1}+ax+b$}.
\newblock {\em Finite Fields their Appl.}, 10(3):285--305, 2004.

\bibitem{bonini2022saturating}
M.~Bonini, M.~Borello, and E.~Byrne.
\newblock Saturating systems and the rank-metric covering radius.
\newblock {\em J. Algebr. Comb.}, 58(4):1173--1202, 2023.

\bibitem{csajbok2017maximum}
B.~Csajb{\'o}k, G.~Marino, O.~Polverino, and F.~Zullo.
\newblock Maximum scattered linear sets and mrd-codes.
\newblock {\em J. Algebr. Comb.}, 46:517--531, 2017.

\bibitem{CMPZ2019}
B.~Csajb{\'o}k, G.~Marino, O.~Polverino, and F.~Zullo.
\newblock A characterization of linearized polynomials with maximum kernel.
\newblock {\em Finite Fields their Appl.}, 56:109--130, 2018.

\bibitem{CMPZ2021-generalising}
B.~Csajb\'{o}k, G.~Marino, O.~Polverino, and F.~Zullo.
\newblock Generalising the scattered property of subspaces.
\newblock {\em Combinatorica}, 41(2):237--262, 2021.

\bibitem{Delsarte1978bilinear}
P.~Delsarte.
\newblock Bilinear forms over a finite field, with applications to coding
  theory.
\newblock {\em J. Combin. Theory Ser. A}, 25(3):226--241, 1978.

\bibitem{Dickson1906}
L.~E. Dickson.
\newblock Criteria for the irreducibility of functions in a finite field.
\newblock {\em Bull. Am. Math. Soc.}, 13:1--8, 1906.

\bibitem{dvir2012subspace}
Z.~Dvir and S.~Lovett.
\newblock Subspace evasive sets.
\newblock In {\em Proceedings of the forty-fourth annual ACM symposium on
  Theory of computing}, pages 351--358, 2012.

\bibitem{Gabidulin1985theory}
E.~M. Gabidulin.
\newblock Theory of codes with maximum rank distance.
\newblock {\em Probl. Peredachi Inf.}, 21(1):3--16, 1985.

\bibitem{GruicaRavagnaniSheekeyZullo2022}
A.~Gruica, A.~Ravagnani, J.~Sheekey, and F.~Zullo.
\newblock Generalised evasive subspaces.
\newblock {\em arXiv:2207.01027v2}.

\bibitem{guruswami2011linear}
V.~Guruswami.
\newblock Linear-algebraic list decoding of folded \textnormal{Reed-Solomon}
  codes.
\newblock In {\em 2011 IEEE 26th Annual Conference on Computational
  Complexity}, pages 77--85. IEEE, 2011.

\bibitem{GuruswamiWangXing2016}
V.~Guruswami, C.~Wang, and C.~Xing.
\newblock Explicit list-decodable rank-metric and subspace codes via subspace
  designs.
\newblock {\em IEEE Trans. Inf. Theory}, 62(5):2707--2718, 2016.

\bibitem{Honold2018}
T.~Honold, M.~Kiermaier, and S.~Kurz.
\newblock Partial {Spreads} and {Vector} {Space} {Partitions}.
\newblock In {\em Network {Coding} and {Subspace} {Designs}}, pages 131--170.
  Springer International Publishing, Cham, 2018.

\bibitem{bok:MW}
F.~MacWilliams and N.~Sloane.
\newblock {\em The Theory of Error-Correcting Codes}.
\newblock North-holland Publishing Company, 2nd edition, 1978.

\bibitem{MarinoNeriTrombetti2023}
G.~Marino, A.~Neri, and R.~Trombetti.
\newblock Evasive subspaces, generalized rank weights and near mrd codes.
\newblock {\em Discrete Math.}, 346(12):113605, 2023.

\bibitem{massey1993minimal}
J.~L. Massey.
\newblock Minimal codewords and secret sharing.
\newblock In {\em Proceedings of the 6th joint Swedish-Russian international
  workshop on information theory}, pages 276--279, 1993.

\bibitem{massey1995some}
J.~L. Massey.
\newblock Some applications of coding theory in cryptography.
\newblock {\em Codes and Ciphers: Cryptography and Coding IV}, pages 33--47,
  1995.

\bibitem{SheekeyMcGuire2019}
G.~McGuire and J.~Sheekey.
\newblock A characterization of the number of roots of linearized and
  projective polynomials in the field of coefficients.
\newblock {\em Finite Fields their Appl.}, 57:68--91, 2019.

\bibitem{patil2021trace}
K.~Pat{\.i}l.
\newblock On the trace of powers of square matrices.
\newblock {\em Hacet. J. Math. Stat.}, 50(1):14--23, 2021.

\bibitem{PudlakVojtvech2004}
P.~Pudl\'{a}k and V.~R\"{o}dl.
\newblock Pseudorandom sets and explicit constructions of {R}amsey graphs.
\newblock In {\em Complexity of computations and proofs}, volume~13 of {\em
  Quad. Mat.}, pages 327--346. Dept. Math., Seconda Univ. Napoli, Caserta,
  2004.

\bibitem{Randrianarisoa2020geometric}
T.~H. Randrianarisoa.
\newblock A geometric approach to rank metric codes and a classification of
  constant weight codes.
\newblock {\em Des. Codes Cryptogr.}, 88(7):1331--1348, 2020.

\bibitem{SheekeyVoorde2020duality}
J.~Sheekey and G.~Van~de Voorde.
\newblock Rank-metric codes, linear sets, and their duality.
\newblock {\em Des. Codes Cryptogr.}, 88(4):655--675, 2020.

\bibitem{silva2008rank}
D.~Silva, F.~R. Kschischang, and R.~Koetter.
\newblock A rank-metric approach to error control in random network coding.
\newblock {\em IEEE Trans. Inf. Theory}, 54(9):3951--3967, 2008.

\bibitem{Collisions}
J.~von~zur Gathen, M.~Giesbrecht, and K.~Ziegler.
\newblock Composition collisions and projective polynomials.
\newblock {\em ISSAC}, pages 123--130, 2010.

\end{thebibliography}

\newpage

\noindent Stefano Lia\\
School of Mathematics and Statistics\\ University College Dublin\\Dublin\\
\texttt{ stefano.lia@ucd.ie}

\bigskip

\noindent Giovanni Longobardi, Giuseppe Marino, Rocco Trombetti\\
Department of Mathematics and its Applications ``Renato Caccioppoli",\\
University of Naples Federico II,\\
Via Cintia, Monte S.Angelo I-80126 Napoli, Italy\\
\texttt{\{giovanni.longobardi, giuseppe.marino, rtrombet\}@unina.it}

\end{document}